\documentclass[amsmath,amssymb,twoside,epstopdf,12pt]{article}
\usepackage{amsfonts}
\usepackage{graphicx}
\usepackage{subfigure}
\usepackage{setspace}
\usepackage{subfig}
\usepackage{array}
\usepackage[margin=1in]{geometry}
\usepackage{amsmath, amsfonts, amssymb, amsthm}
\newtheorem{thm}{Theorem}[section]

\newtheorem{pro}{Proposition}[section]

\theoremstyle{definition}
\newtheorem{defi}{Definition}[section]
\theoremstyle{remark}

\numberwithin{equation}{section} \hfuzz3pt
\begin{document}
\setcounter{page}{1} \vspace{0.1cm}
\begin{center}
{\large{\bf Multi-period investment strategies under Cumulative Prospect Theory   }}\\
 \vspace{0.4cm}
{\bf Liurui Deng$^1$ and Traian A. Pirvu$^2$}\\
$^1$College of Economics and Management, Hunan Normal University,\\ Changsha, 410081, China\\
purplerosed@yahoo.com\\
$^2$ Department of Mathematics and Statistics, McMaster University\\
1280 Main Street West, Hamilton, ON, L8S 4K1\\
tpirvu@math.mcmaster.ca\\
 \vspace{0.4cm}

\end{center}
%

\vspace{0.2cm}


\begin{abstract}

In this article, inspired by Shi, et al. we investigate the
optimal portfolio selection with one risk-free asset and
one risky asset in a multiple period setting under cumulative
 prospect theory (CPT). Compared with their study, our novelty
  is that we consider a stochastic benchmark, and portfolio constraints.
   We test the sensitivity of the optimal CPT-investment strategies to different
  model parameters by performing a numerical analysis.
\end{abstract}



\section{Introduction}

Expected utility theory (EUT) was a popular decision making model in finance and economics. In reality, however, various decision makers' behaviour deviates from the implications
 of expected utility. Substantial experimental and empirical evidence points out that expected utility theory is
 incompatible with human observed behaviour. The abundant paradoxes lead to the development of a more realistic theory.
 One of them is prospect theory (PT) (see \cite{Tversky1}). Later, PT extends to cumulated prospect theory (CPT) because
 CPT is consistent with the first-order stochastic dominance (see \cite{Tversky2}).

Let us go over some of the literature on CPT. In a continuous-time setting \cite{Jin} formulate a general behavioural portfolio
selection model under Kahneman and Tversky's cumulative prospect theory. In a discrete-time setting \cite{Bernard1} considers
 how a CPT investor chooses his/her optimal portfolio in a single period setting with one risky and one riskless asset. In the same
 vein \cite{He1} and \cite{He2} address and formulate the well-posedness of CPT criterion and investigate the case in which the
 reference point is different from the risk-free return. The extension to a multi asset paradigm was done in \cite{Pirvu} and \cite{kwak}. To the best
 of our knowledge \cite{Car} and \cite{Shi} were the only papers to consider the CPT allocation problem in a multi period framework.

In this  paper, we study the optimal portfolio of a CPT investor. Inspired by \cite{Shi} we focus on the allocation problem with one risk-free
 asset and one risky asset in the multiple periods setting and CPT risk criterion. \cite{Shi} points out that the optimal strategies are time consistent in their setting due to probability distortions. 
 Compared with \cite{Shi} one of the contributions of the present article is to allow for time changing benchmark and portfolio constraints. A time changing benchmark is a more realistic model of portfolio management. The work of \cite{Strub} investigates weather or not reference point updating lead to time-inconsistent investment. They analyze both the time consistent and time inconsistent frameworks and conclude that the updating of the reference point non recursively leads to realistic trading behaviour which are time inconsistent and it embeds the disposition effect.

 We have to face the time inconsistency of optimal investment strategies in our setting. Due to this predicament we consider a special type of benchmarking which renders optimal strategies time consistent. The
 optimal strategy is computed by backward induction (see e.g. \cite{Kwak} and \cite{Pirvu1} for similar techniques). Our main result is a recursive formula to characterize the optimal strategies. The stochastic model we consider for the stock return is fairly general. One interesting feature of our recursive formulas characterizing the optimal
strategies is that they allow for any distribution specification of stock returns. This is explained by the fact that optimal strategies are computed by
backward induction. Our numerical experiments reveal the effect of time on the optimal strategies (this effect can not be observed in one period models).
The main numerical finding in this regard is that the effect of the model parameters on optimal strategies diminishes as time goes by and it gets closer to maturity.

 The remainder of this paper  is organized as follows: In Section 2 we present the model. Section 3 outlines the objective. The results are presented in Section 4. Numerical analysis is performed in Section 5. The paper ends with an Appendix containing the proofs.

\section{The Model}
We consider a financial market with two investment opportunities: one risk-free asset and one risky asset. The investment horizon is $[0,T],$
where $T$ is a finite deterministic positive constant. Let $t$ denote the time of investment which takes on discrete values $(t=0,1,...,T-1).$ The future evolution of risky assets's price is modelled on a probability space $(\Omega, \mathcal{F}_t, \mathcal{F}, P)$. The information set at the beginning of period $t$ is $\mathcal{F}_t.$ Let $r_t$ denote the return of the risk-free asset from time $t$ to time $t+1$, and let $x_t$ denote the return of the risky asset from time $t-1$ to time $t$. We assume that an investor has wealth $W_t$ at time $t$ and that he/she invests the amount $v_t$ in the risky asset and all of remaining wealth $W_t-v_t$ in the risk-free asset. The investor's wealth $W_{t+1}$ at time $t+1$ is given by the self-financing equation
\begin{eqnarray}\label{eq14}
W_{t+1}&=&(1+r_t)(W_{t}-v_t)+v_t(1+x_{t+1})\nonumber\\
&=&(1+r_t)W_t+v_t(x_{t+1}-r_t)\nonumber\\
&=&(1+r_t)W_t+v_t y_{t+1}.
\end{eqnarray}
Here $y_{t+1}$ is a random variable which represents the dollar excess return on the risky asset over the risk-free asset from time $t$ to time $t+1$.
The excess return process $\{y_t\}_{\{ t=0,1,...,T-1\}},$ the interest rate process $\{r_t\}_{\{ t=0,1,...,T-1\}},$ and the amount invested in the risky asset process $\{v_t\}_{\{ t=0,1,...,T-1\}}$ are adapted stochastic processes, i.e., $y_t \in \mathcal{F}_{t},$ $r_t \in \mathcal{F}_{t},$ and $v_t \in \mathcal{F}_{t}.$

\subsection{The Benchmarked Wealth}
Investors asses the performance of their investments in comparison to a benchmark. This approach is modelled by a benchmark process, which at any point in time divides 
investment outcomes into gains and losses as follows: a gain occurs whenever the investment outcomes are above the benchmark, and a loss occurs whenever the investment outcomes are below the benchmark.

Let $B_t, 0\leq t \leq T,$ be a positive benchmark process. The benchmarked wealth at time $t+1,$ given initial time $t$ is
\begin{equation}\label{eq1}
\overline {W}^{t+1}_{t}= {W}^{t+1}_{t}- B_{t+1}.
\end{equation}
Given the initial time $t$ the benchmarked wealth at time $t+2$ is
\begin{equation}\label{eq12}
\overline{W}^{t+2}_{t}=W^{t+2}_{t}-B_{t+2}
\end{equation}
 In general, the benchmarked wealth is given by 
\begin{equation}\label{1e1}
\overline{W}_{t}^s={W}_{t}^s-B_{s}, s=t+1,t+2,...,T.
\end{equation}
\subsubsection{Examples of Benchmark}
One benchmark process is the constant throughout time benchmark $B,$ where $B$ is a constant. This kind of benchmark is considered in \cite{Shi}. 
Another benchmark process is obtained by investing the initial/current wealth in the risk free asset. Other examples of benchmarks are an index of stochastic assets, and the expected
wealth (for more on this see \cite{Pirvu}). 

\subsection{Portfolio Constraints}
It is often the case that managers impose risk limits on the trading strategies. The risk control mechanism has two components: one internal (imposed by risk management departments) and one external
( accredited regulatory institutions). Thus, it is only natural to consider portfolio constraints. In the following we introduce the class of admissible strategies.
\begin{defi}
The set of admissible strategies at time $t$ is $\mathcal{{A}}_{t},$
\begin{equation}\label{const}
\mathcal{{A}}_{t}= \{v_{t}\in \mathcal{{F}}_{t} | \quad A |W_{t}|\leq v_{t}\leq B |W_{t}| \}
\end{equation}
for some constants $A\leq 0<B.$
\end{defi}
This formal definition is saying that the investments in the risky asset are acceptable if the corresponding portfolio proportion invested in the risky asset is bounded.

The fact that  $A\leq 0$ is saying that shortselling is allowed in this market. In practice there are rules about shortselling; for the special case of Chinese financial markets: 
"The underlying stocks for shortselling should have more than 200 million shares or the market value in circulation should have more than 800 million RMB; the number of
 shareholders should be more than 4000," conform page 263 in \cite{G}.

The portfolio constraint $A |W_{t}|\leq v_{t}$ sets a (negative floor) on the portfolio proportion invested in the risky asset, which is the same as limiting the shortselling amount
(some portfolio managers tend to do that). On the other hand, the constraint $v_{t}\leq B |W_{t}|$ limits the exposure to the risky asset, which is in line with the industry practice: "In 1991, the National Association of Insurance Commissioners (NAIC) imposed higher reserve requirements on insurance companies holdings of junk bonds, specifying a 20$\%$ cap on the assets insurers may hold in junk bonds. Many pension funds place limits on the fraction of a portfolio that can be invested in junk bonds," conform \cite{DG} (junk bonds are a special type of risky assets which have a high return, but also posses a high risk).

\subsection{The CPT Risk Criterion}

The investor gets utility from gains and disutility from the losses. The benchmark differentiates gains from losses.
Let us introduce the following formal definitions.

\begin{defi}(see \cite{Tversky1} and \cite{Tversky2}) The value function $u$ is defined as follows:
\[u(x)=\left\{\begin{array}{ll}
u^+(x)&\text{if $x\geq 0$ },\\
-u^-(-x)&
\text{if $x<0$},
\end{array}\right.\]

where $u^+:\overline{\mathbb{R}}^+\rightarrow \overline{\mathbb{R}}^+$ and $u^-:\overline{\mathbb{R}}^+\rightarrow\overline{\mathbb{R}}^+$ satisfy:

(i) $u(0)=u^+(0)=u^-(0)=0;$

(ii)$u^+(+\infty)=u^-(+\infty)=+\infty;$

(iii)$u^+(x)=x^\alpha,$ with$ 0<\alpha<1$ and $x\geq 0;$

(iv)$u^-(x)=\lambda x^\alpha$ with $\lambda>1,$ and $x\geq 0.$
\end{defi}

It is important to point out that the CPT objective considers deviation from benchmark, so the argument $x$ in the utility function $u(x)$ is not a wealth level but a deviation of wealth from benchmark.

\begin{defi}
Let $F_{X}(\cdot)$ be the cumulative distribution function (CDF) of a random variable $X$. The probability distortions are denoted by $T^+$ and $T^-$. We define the two probability weight functions (distortions) $T^+: [0,1] \rightarrow [0,1]$ and $T^-: [0,1] \rightarrow [0,1]$ as follows:
\begin{eqnarray*}
&&T^+(F_{X}(x))=\frac{F_{X}^\gamma(x)}{(F_{X}^\gamma(x)+(1-F_{X}(x))^\gamma)^{1/\gamma}}, ~~\text{with}~~0.28<\gamma<1,\\
&&T^-(F_{X}(x))=\frac{F_{X}^\delta(x)}{(F_{X}^\delta(x)+(1-F_{X}(x))^\delta)^{1/\delta}}, ~~\text{with}~~0.28<\delta<1.
\end{eqnarray*}
\end{defi}

\begin{defi}\label{def1}
Define the objective function of the CPT-investor, denoted by $U(W)$, as:
\begin{eqnarray}
U(W)=\int_0^{+\infty}T^+(1-F_{W}(x))du^+(x)-\int_0^{+\infty}T^-(F_{W}(-x))du^-(x),
\end{eqnarray}
where $W$ is the benchmarked wealth.
$U(W)$ is a sum of two Choguet integrals (see \cite{Choquet} and \cite{Chateauneuf}). It is well-defined when
$$\alpha<2\min(\delta,\gamma),$$
(see Proposition 2 in \cite{Pirvu}).
\end{defi}
From Definition \ref{def1} it follows that
an objective function of the CPT-investor at the time $t$ is:
\begin{eqnarray}
U(\overline{W}_{t}^T)=\int_0^{+\infty}T^+(1-F_{\overline{W}_{t}^T}(x))du^+(x)-\int_0^{+\infty}T^-(F_{\overline{W}_{t}^T}(-x))du^-(x).
\end{eqnarray}

As we mentione in Subsection \ref{inc} this risk criterion is time inconsistent, thus we propose a special case of benchmarked wealth. That is, perform benchmarking at time $T-1$ only\footnote{This assumption is made to render the portfolio problem tractable.}. Thus, if the initial time is $t$ ($t=0,1,2,...,T-1$), then the benchmarked wealth at $T$ is:
\begin{equation}\label{eu1}
\underline{W}_{t}^T=v_{T-1}y_{T}\quad\mbox{given that the initial time is}\,\,\,t.\end{equation}

The objective function is then defined as follows
\begin{equation}\label{!1}
U(\underline{W}_{T-1}^T)=\int_0^{+\infty}T^+(1-F_{\underline{W}_{T-1}^T}(x))du^+(x)-\int_0^{+\infty}T^-(F_{\underline{W}_{T-1}^T}(-x))du^-(x),
\end{equation}
and
\begin{equation}\label{!2}
U(\underline{W}_{t}^T)=E[U(\underline {W}_{T-1}^T) |\mathcal{F}_{t}].
\end{equation}
The above equation naturally extends $U(\underline{W}_{T-1}^T)$ at earlier times and this makes sense in light of \eqref{eu1}.  

\section{Objective}
In this section we formulate the CPT investor objective. In a first step we consider the one period model.

\subsection{Single Period  Objective}
We assume that the current time is $T-1$. Recall that the benchmarked wealth at time $T$ is
\begin{eqnarray*}
\underline{W}_{T-1}^T=v_{T-1} y_{T} .
\end{eqnarray*}

The CPT-investor objective is to find the portfolio which leads to the highest possible prospect value; consequently, the investor solves the following portfolio problem
\begin{eqnarray}
&(P)&\nonumber\\
&&\max_{v_{T-1}\in \mathcal{{A}}_{T-1}}U(\underline{W}_{T-1}^T)\nonumber\\
&=&\max_{v_{T-1}\in  \mathcal{{A}}_{T-1}}\Big[\int_0^{+\infty}T^+(1-F_{\underline{W}_{T-1}^T}(x))du^+(x)-\int_0^{+\infty}T^-(F_{\underline{W}_{T-1}^T}(-x))du^-(x)\Big].\nonumber\\
\end{eqnarray}
The portfolio yielding the highest possible prospect value, i.e., the optimal portfolio, is $v_{T-1}^*$ given by
\begin{eqnarray}
v_{T-1}^*&=&arg \max_{v_{T-1}\in  \mathcal{{A}}_{T-1}}U(\underline{W}_{T-1}^T)\nonumber\\
&=&arg\max_{v_{T-1}\in  \mathcal{{A}}_{T-1}}\Big[\int_0^{+\infty}T^+(1-F_{\underline{W}_{T-1}^T}(x))du^+(x)-\int_0^{+\infty}T^-(F_{\underline{W}_{T-1}^T}(-x))du^-(x)\Big].\nonumber\\
\end{eqnarray}

\subsection{Multiple Periods Objective}
Let us consider a multiple periods model. In order to get a time consistent optimal strategies we work with the benchmarked wealth $\underline{W}_{t}^T,$ defined in \eqref{eu1}. The CPT-investor objective is to solve the following portfolio problem
\begin{eqnarray}
&(P)&\nonumber\\
&&\max_{v_{i}\in \mathcal{{A}}_{i},\, i=t,t+1,\ldots T-1}U(\underline{W}_{t}^T)\nonumber\\
&=&\max_{v_{i}\in  \mathcal{{A}}_{i},\, i=t,t+1,\ldots T-1} E \Big[\int_0^{+\infty}T^+(1-F_{\underline{W}_{T-1}^T}(x))du^+(x)-\int_0^{+\infty}T^-(F_{\underline{W}_{T-1}^T}(-x))du^-(x) \Big|  \mathcal{{F}}_{t} \Big].\nonumber\\
\end{eqnarray}
Having stated the objectives let us turn to the issue of time inconsistency.

\subsection{Time Inconsistency}\label{inc}
As already pointed out in \cite{Shi} the optimal strategy is time inconsistent due to the nonlinear distortions. That is the optimal strategy computed in the past is not implemented unless there is a commitment mechanism. Let us define $v_{T-1}^*$ the time $T-1$ portfolio and $(\hat{v}_{T-2},\hat{v}_{T-1})$ the time $T-2$ optimal portfolio:
$$ v_{T-1}^*=arg \max_{v_{T-1}\in  \mathcal{{A}}_{T-1}}U(\underline{W}_{T-1}^T),$$ and
$$ (\hat{v}_{T-2},\hat{v}_{T-1})=arg \max_{v_{T-2}\in\mathcal{{A}}_{T-2}, v_{T-1}\in\mathcal{{A}}_{T-1}}U(\underline{W}_{T-2}^T).$$
It turns out that $v_{T-1}^*\neq \hat{v}_{T-1},$ which means that the investor is time inconsistent, i.e., the time $T-2$ optimal portfolio is not implemented at time $T-1$ when another portfolio
is optimal.

\section{Results}

After formulating the objective we are ready to present our results. Our exposition starts with the single period case.
\subsubsection{One Period Model}
We manage to analyze the investor portfolio optimization problem in terms of the wealth being positive/negative.
The result is summarized in the following Proposition. 
\begin{pro}\label{1}
 The optimal prospect value is 
\begin{eqnarray}
\max_{v_{T-1}\in  \mathcal{{A}}_{T-1}   }U(\underline{W}_{T-1}^T)=
W_{T-1}^\alpha A_{T-1}I_{W_{T-1}\geq 0}-(-W_{T-1})^\alpha B_{T-1} I_{W_{T-1}<0},\nonumber\\
\end{eqnarray}
where
\begin{eqnarray}\label{t1}
A_{T-1}=\max_{z\in[A,B]}g_{T-1}(z),
\end{eqnarray}
\begin{eqnarray}\label{t2}
B_{T-1}=-\max_{z\in[-B,-A]}l_{T-1}(z),
\end{eqnarray}
\begin{eqnarray*}
g_{T-1}(z)=z^\alpha k(T-1)I_{z\in[0,B]}+(-z)^\alpha h(T-1)I_{z\in[A,0]}
\end{eqnarray*}
\begin{eqnarray*}
l_{T-1}(z)=(-z)^\alpha k(T-1)I_{z\in[-B,0]}+z^\alpha h(T-1)I_{z\in[0,-A]}
\end{eqnarray*}
\begin{eqnarray*}
k(T-1)=\int_0^{+\infty}T^+(1-F_{y_{T}}(x))du^+(x)-\int_0^{+\infty}T^-(F_{y_{T}}(-x))du^-(x)
\end{eqnarray*}
and
\begin{eqnarray*}
h_(T-1)=\int_0^{+\infty}T^+(1-F_{-y_{T}}(x))du^+(x)-\int_0^{+\infty}T^-(F_{-y_{T}}(-x))du^-(x)
\end{eqnarray*}
\end{pro}

\begin{proof}
The proof is given in Appendix A.
\end{proof}

\begin{thm}\label{th111}
The optimal CPT-investment strategy is
\[v_{T-1}^*=\left\{\begin{array}{ll}
k_{T-1}^*W_{T-1}&\text{if $W_{T-1}\geq 0$ },\\
\hat{k}^*_{T-1}W_{T-1}&
\text{if $W_{T-1}<0$},
\end{array}\right.\]
where
\[k_{T-1}^*=arg\max_{z\in[A,B]}g_{T-1}(z),\qquad \hat{k}_{T-1}^*=arg\max_{z\in[-B,-A]}l_{T-1}(z).\]
\end{thm}

\begin{proof}
The above conclusion is derived from Proposition \ref{1}.
\end{proof}

\subsubsection{Multiple Periods Model }
Subsequently, we consider the optimal strategies in the multiple periods. In order to ease the exposition we present the two period model case; when the initial time is $T-2,$ the following result is established.
\begin{pro}\label{th2}
The maximum CPT value is given recursively by
\begin{eqnarray}
 \max_{v_{T-2}\in\mathcal{{A}}_{T-2}, v_{T-1}\in\mathcal{{A}}_{T-1}} U(\underline{W}_{T-2}^T)&=&A_{T-2}W_{T-2}^\alpha I_{W_{T-2}\geq0}-B_{T-2}(-W_{T-2})^\alpha I_{W_{T-2}<0},\nonumber\\
\end{eqnarray}
where
\begin{equation}\label{tt1}
A_{T-2}=\max_{z\in[A,B]} g_{T-2}(z),
\end{equation}
\begin{eqnarray}\label{tt2}
B_{T-2}=-\max_{z\in[-B,-A]}l_{T-2}(z),
\end{eqnarray}
\begin{eqnarray*}
&&g_{T-2}(z)=E[A_{T-1}(1+r_{T-2}+y_{T-1}z)^\alpha I_{\{1+r_{T-2}+y_{T-1}z\geq0\}}\\
&-&B_{T-1}(-1-r_{T-2}-y_{T-1}z)^\alpha I_{\{1+r_{T-2}+y_{T-1}z<0\}}|\mathcal{F}_{T-2}]
\end{eqnarray*}
\begin{eqnarray*}
&&l_{T-2}(z)=E[A_{T-1}(-1-r_{T-2}-y_{T-1}z)^\alpha I_{\{1+r_{T-2}+y_{T-1}z<0\}}\\
&-&B_{T-1}(1+r_{T-2}+y_{T-1}z)^\alpha I_{\{1+r_{T-2}+y_{T-1}z\geq0\}}|\mathcal{F}_{T-2}].
\end{eqnarray*}
and $A_{T-1}, B_{T-1}$ are given by \eqref{t1}, \eqref{t2}.
\end{pro}
\begin{proof}
The proof is done in Appendix B.
\end{proof}

As a consequence of this we can obtain the following key Proposition and Theorem for multiple periods.
They provide an algorithm to compute iteratively backward (starting from the last time period $[T-1,T]$ and then going to $[T-2,T],$ etc)
the optimal prospect value and the corresponding optimal portfolio. 
\begin{pro}\label{th3}
Given the initial time $t,$ the optimal CPT value is given recursively by
\begin{eqnarray}
 \max_{v_{i}\in \mathcal{{A}}_{i},\, i=t,t+1,\ldots T-1}U(\underline{W}_{t}^T)&=&A_{t}W_{t}^\alpha I_{W_{t}\geq0}-B_{t}(-W_{t})^\alpha I_{W_{t}<0},
\end{eqnarray}
where
$$A_{t}=\max_{z\in[A,B]} g_{t}(z),$$
\begin{eqnarray*}
B_{t}=-\max_{z\in[-B,-A]}l_{t}(z),
\end{eqnarray*}
\begin{eqnarray*}
&&g_{t}(z)=E[A_{t+1}(1+r_{t}+y_{t+1}z)^\alpha I_{\{1+r_{t}+y_{t+1}z\geq0\}}\\
&-&B_{t+1}(-1-r_{t}-y_{t+1}z)^\alpha I_{\{1+r_{t}+y_{t+1}z<0\}}|\mathcal{F}_{t}]
\end{eqnarray*}
\begin{eqnarray*}
&&l_{t}(z)=E[A_{t+1}(-1-r_{t}-y_{t+1}z)^\alpha I_{\{1+r_{t}+y_{t+1}z<0\}}\\
&-&B_{t+1}(1+r_{t}+y_{t+1}z)^\alpha I_{\{1+r_{t}+y_{t+1}z\geq0\}}|\mathcal{F}_{t}],
\end{eqnarray*}
$A_{T-1}, B_{T-1}$ are given by \eqref{t1}, \eqref{t2}, $A_{T-2}, B_{T-2}$ are given by \eqref{tt1}, \eqref{tt2}, etc.
\end{pro}
\begin{proof}
We prove this in Appendix C.
\end{proof}

The following Theorem is our main result of the paper.

\begin{thm}\label{th4}
The optimal CPT-investment strategy $(v_{0}^*,v_1^*,...,v_T^*)$ is given recursively for $t=T-1,T-2,\ldots,0$ by
\[v_{t}^*=\left\{\begin{array}{ll}
k_{t}^*W_{t}&\text{if $W_{t}\geq 0$ },\\
\hat{k}^*_{t}W_{t}&
\text{if $W_{t}<0$},
\end{array}\right.\]
\begin{eqnarray*}
k_{t}^*=arg\max_{z\in[A,B]}g_{t}(z)
\end{eqnarray*}
and
\begin{eqnarray*}
\hat{k}_{t}^*=arg\max_{z\in[-B,-A]}l_{t}(z).
\end{eqnarray*}
\end{thm}

\begin{proof}
The above result is derived from Proposition \ref{th3}.
\end{proof}

\section{Numerical Analysis}

\subsection{Numerical simulation}
We suppose the excess return $y_t$ satisfies a normal distribution $N(\mu, \sigma)$. The interest rate is set to follow a Ho and Lee model
$$r_t=0.03 +0.003\sqrt{t} Z,\qquad Z\sim N(0,1).$$

 We test the sensitivity of optimal solution to different parameters $\alpha$, $\mu$ and $\delta$. Moreover,  we analyze how CPT-investors'  psychology and the characteristics of the stock influence the optimal CPT-investment strategies. When we test sensitivity, we set $W_0=0.8$ and discuss sensitivity of t-th period (t=0,1,...,10). Set $\lambda=2.20,\delta=0.69, A=-5, B=5.$ Let $\mu=0.045, \sigma=1.69$ in Figure 1 and Figure 2 and let $\alpha=0.88$ in Figure 3 and Figure 4. Firstly, in order to reveal the effect of CPT-investors' psychology on the optimal portfolio choice, we analyze the sensitivity of the optimal solution to the parameter $\alpha$ (see Figure 1). The ratio invested in the risky asset is increasing in $\alpha.$ Moreover it is slowly decreasing in time. As expected the investment in the risky asset is decreasing in $\sigma$ and increasing in $\mu.$ Moreover it turns out to be slowly decreasing in time. It interesting to point out that, conform Figure 5 and Figure 6, a random interest rate forces a higher investment in the risky asset, fact explained by the risk incurred in investing in the risk free asset.
\vspace{0.5cm}

{\bf Acknowledgement} The article is  supported by National Natural Science Foundation of China (71201051),
Young Talents Training Plan of Hunan Normal University(2014YX04), Philosophical and Social Science Fund of Hunan (No.14YBA264), and NSERC grant 396-5058-14.
These fundings did not lead to any conflict of interests regarding the publication of this manuscript.

\section{Appendix}

\subsection{Appendix A: proof of Proposition \ref{1}}
\begin{proof}
When $W_{T-1}\geq 0$ and $0\leq v_{T-1}\leq B W_{T-1} $, we directly employ the result of \cite{Bernard1} in order to get
\begin{eqnarray}
&&U(\underline{W}_{T-1}^T)\nonumber\\
&=&v_{T-1}^\alpha\Big(\int_0^{+\infty}T^+(1-F_{y_{T}}(x))du^+(x)-\int_0^{+\infty}T^-(F_{y_{T}}(-x))du^-(x)\Big).\nonumber\\
\end{eqnarray}

Let $v_{T-1}=W_{T-1}z,$ so
\begin{eqnarray}
&&U(\underline{W}_{T-1}^T)\nonumber\\
&=&W_{T-1}^\alpha z^\alpha\Big(\int_0^{+\infty}T^+(1-F_{y_{T}}(x))du^+(x)-\int_0^{+\infty}T^-(F_{y_{T}}(-x))du^-(x)\Big).\nonumber\\
&=&W_{T-1}^\alpha z^\alpha k(T-1)
\end{eqnarray}

Similarly, when $W_{T-1}\geq 0$ and $A W_{T-1}\leq v_{T-1}\leq 0 $, we can show that
\begin{eqnarray}
&&U(\underline{W}_{T-1}^T)\nonumber\\
&=&(-v_{T-1})^\alpha\Big(\int_0^{+\infty}T^+(1-F_{-y_{T}}(x))du^+(x)-\int_0^{+\infty}T^-(F_{-y_{T}}(-x))du^-(x)\Big)\nonumber\\
&=&W_{T-1}^\alpha(-z)^\alpha\Big(\int_0^{+\infty}T^+(1-F_{-y_{T}}(x))du^+(x)-\int_0^{+\infty}T^-(F_{-y_{T}}(-x))du^-(x)\Big)\nonumber\\
&=&W_{T-1}^\alpha(-z)^\alpha h(T-1).
\end{eqnarray}

Hence, when $W_{T-1}\geq 0$, we have that
\begin{eqnarray}
U(\underline{W}_{T-1}^T)&=&W_{T-1}^\alpha[ z^\alpha k(T-1)I_{z\in [0,B]}+(-z)^\alpha h(T-1)I_{z\in [A,0]}]\nonumber\\
&=&W_{T-1}^\alpha g_{T-1}(z).
\end{eqnarray}

When $W_{T-1}<0$ and $0\leq v_{T-1}\leq -B W_{T-1} $, one easily gets
\begin{eqnarray}
&&U(\underline{W}_{T-1}^T)\nonumber\\
&=&(-W_{T-1})^\alpha(-z)^\alpha\Big(\int_0^{+\infty}T^+(1-F_{y_{T}}(x))du^+(x)-\int_0^{+\infty}T^-(F_{y_{T}}(-x))du^-(x)\Big)\nonumber\\
&=&(-W_{T-1})^\alpha(-z)^\alpha k(T-1)
\end{eqnarray}
When $W_{T-1}<0$ and $-A W_{T-1}\leq v_{T-1}\leq 0$, we show that
\begin{eqnarray}
&&U(\underline{W}_{T-1}^T)\nonumber\\
&=&(-v_{T-1})^\alpha\Big(\int_0^{+\infty}T^+(1-F_{-y_{T}}(x))du^+(x)-\int_0^{+\infty}T^-(F_{-y_{T}}(-x))du^-(x)\Big)\nonumber\\
&=&(-W_{T-1})^\alpha z^\alpha\Big(\int_0^{+\infty}T^+(1-F_{-y_{T}}(x))du^+(x)-\int_0^{+\infty}T^-(F_{-y_{T}}(-x))du^-(x)\Big)\nonumber\\
&=&(-W_{T-1})^\alpha z^\alpha h(T-1).
\end{eqnarray}

Therefore, when $W_{T-1}<0$, we have that
\begin{eqnarray}
U(\underline{W}_{T-1}^T)&=&(-W_{T-1})^\alpha[(-z)^\alpha k(T-1)I_{z\in [-B,0]}+z^\alpha h(T-1)I_{z\in [0,-A]}]\nonumber\\
&=&(-W_{T-1})^\alpha l_{T-1}(z).
\end{eqnarray}

\end{proof}
\subsection{Appendix B: proof of Proposition \ref{th2}}
\begin{proof}
Since
\begin{eqnarray}
\underline{W}_{T-2}^T=v_{T-1}y_{T}\quad\mbox{given that the initial time is}\,\,\,T-2,\end{eqnarray}
and
$$
U(\underline {W}_{T-2}^T)=E[U(\underline {W}_{T-1}^T) |\mathcal{F}_{T-2}],
$$
then
$$
\max_{v_{T-2}\in\mathcal{A}_{T-2},v_{T-1}\in\mathcal{A}_{T-1} }U(\underline {W}_{T-2}^T)=\max_{v_{T-2}\in\mathcal{A}_{T-2}}E[\max_{v_{T-1}\in\mathcal{A}_{T-1} }U(\underline {W}_{T-1}^T) |\mathcal{F}_{T-2}].
$$
By applying Proposition \ref{1} to the above equation, one gets
\begin{eqnarray}
&&\max_{v_{T-2}\in\mathcal{A}_{T-2},v_{T-1}\in\mathcal{A}_{T-1}}U(\underline {W}_{T-2}^T)\nonumber\\
&=&\max_{v_{T-2}\in\mathcal{A}_{T-2}}E[W_{T-1}^\alpha A_{T-1}I_{W_{T-1}\geq 0}-(-W_{T-1})^\alpha B_{T-1} I_{W_{T-1}<0}|\mathcal{F}_{T-2}].\nonumber\\
\end{eqnarray}

From equation (\ref{eq14}) it follows that
\begin{eqnarray*}
W_{T-1}=(1+r_{T-2})W_{T-2}+v_{T-2}y_{T-1}.
\end{eqnarray*}
Let $$v_{T-2}=W_{T-2}k_{T-2}.$$
When $W_{T-2}\geq 0$, it is straightforward to show that
\begin{eqnarray}
&& \max_{v_{T-2}\in\mathcal{A}_{T-2},v_{T-1}\in\mathcal{A}_{T-1}}U(\underline {W}_{T-2}^T)\nonumber\\
&=&W_{T-2}^\alpha\max_{k_{T-2}\in[A,B]} E[(1+r_{T-2}+k_{T-2}y_{T-1})^\alpha A_{T-1}I_{1+r_{T-2}+k_{T-2}y_{T-1}\geq 0}\nonumber\\
&-&(-1-r_{T-2}-k_{T-2}y_{T-1})^\alpha B_{T-1} I_{1+r_{T-2}+k_{T-2}y_{T-1}<0}|\mathcal{F}_{T-2}]\nonumber\\
&=&W_{T-2}^\alpha \max_{k_{T-2}\in[A,B]}g_{T-2}(k_{T-2})\nonumber\\
&=&W_{T-2}^\alpha A_{T-2}.
\end{eqnarray}

When $W_{T-2}<0$, similarly, we confirm that
\begin{eqnarray}
&&  U(\underline {W}_{T-2}^T)\nonumber\\
&=&(-W_{T-2})^\alpha\max_{\hat{k}_{T-2}\in[-B,-A]} E[(-1-r_{T-2}-\hat{k}_{T-2}y_{T-1})^\alpha A_{T-1}I_{1+r_{T-2}+\hat{k}_{T-2}y_{T-1}<0}\nonumber\\
&-&(1+r_{T-2}+\hat{k}_{T-2}y_{T-1})^\alpha B_{T-1} I_{1+r_{T-2}+\hat{k}_{T-2}y_{T-1}\geq0} |\mathcal{F}_{T-2}]\nonumber\\
&=&(-W_{T-2})^\alpha\max_{\hat{k}_{T-2}\in[-B,-A]}l_{T-2}(\hat{k}_{T-2})\nonumber\\
&=&-(-W_{T-2})^\alpha B_{T-2}.
\end{eqnarray}
Therefore, Propsition \ref{th2} holds.
\end{proof}
\subsection{Appendix C: proof of Proposition \ref{th3}}
\begin{proof}

We use mathematical induction to prove this proposition.
Proposition \ref{1} and Proposition \ref{th2} display that the conclusion of Proposition \ref{th3} holds at times T-1 and T-2.
We suppose the conclusion holds at time $t+1$. Namely,
\begin{eqnarray}
 \max_{v_{i}\in \mathcal{{A}}_{i},\, i=t+1,t+2\ldots T-1}U(\underline{W}_{t+1}^T)&=&A_{t+1}W_{t+1}^\alpha I_{W_{t+1}\geq0}-B_{t+1}(-W_{t+1})^\alpha I_{W_{t+1}<0}.
\end{eqnarray}

Since
\begin{eqnarray}
\underline{W}_{t}^T=v_{T-1}y_{T}\quad\mbox{given that the initial time is}\,\,\,t,\end{eqnarray}
and
$$
U(\underline {W}_{t}^T)=E[U(\underline {W}_{T-1}^T) |\mathcal{F}_{t}],
$$
then
\begin{eqnarray}
&& \max_{v_{i}\in \mathcal{{A}}_{i},\, i=t,t+1\ldots T-1} U(\underline {W}_{t}^T)\nonumber\\
&=&\max_{v_{t}\in\mathcal{{A}}_{t} }E_{t}[  \max_{v_{i}\in \mathcal{{A}}_{i},\, i=t,t+1\ldots T-1}U(\underline {W}_{t+1}^T)]\nonumber\\
&=&\max_{v_{t} \in\mathcal{{A}}_{t} }E[A_{t+1}W_{t+1}^\alpha I_{W_{t+1}\geq0}-B_{t+1}(-W_{t+1})^\alpha I_{W_{t+1}<0}\Big|\mathcal{F}_{t}].\nonumber\\
\end{eqnarray}
Let $$v_{t}=W_{t}k_{t}.$$
Since
\begin{eqnarray*}
W_{t+1}=(1+r_{t})W_{t}+v_{t}y_{t+1},
\end{eqnarray*}
we have
\begin{eqnarray*}
W_{t+1}=W_{t}(1+r_{t}+k_{t}y_{t+1}).
\end{eqnarray*}

Therefore,
when $W_{t}\geq 0$,
\begin{eqnarray}
&& \max_{v_{i}\in \mathcal{{A}}_{i},\, i=t,t+1\ldots T-1} U(\underline {W}_{t}^T)\nonumber\\
&=&W_{t}^\alpha\max_{k_{t}\in[A,B]} E[(1+r_{t}+k_{t}y_{t+1})^\alpha A_{t+1}I_{1+r_{t}+k_{t}y_{t+1}\geq 0}\nonumber\\
&-&(-1-r_{t}-k_{t}y_{t+1})^\alpha B_{t+1} I_{1+r_{t}+k_{t}y_{t+1}<0}|\mathcal{F}_{t}]\nonumber\\
&=&W_{t}^\alpha \max_{k_{t}\in[A,B]}g_{t}(k_{t})\nonumber\\
&=&W_{t}^\alpha A_{t}.
\end{eqnarray}

When $W_{t}<0$,  we can similarly find that
\begin{eqnarray}
&&\max_{v_{i}\in \mathcal{{A}}_{i},\, i=t,t+1\ldots T-1} U(\underline {W}_{t}^T)\nonumber\\
&=&(-W_{t})^\alpha\max_{\hat{k}_{t}\in[-B,-A]} E[(-1-r_{t}-\hat{k}_{t}y_{t+1})^\alpha A_{t+1}I_{1+r_{t}+\hat{k}_{t}y_{t+1}<0}\nonumber\\
&-&(1+r_{t}+\hat{k}_{t2}y_{t+1})^\alpha B_{t+1} I_{1+r_{t}+\hat{k}_{t}y_{t+1}\geq0}|\mathcal{F}_{t}]\nonumber\\
&=&(-W_{t})^\alpha\max_{\hat{k}_{t}\in[-B,-A]}l_t(\hat{k}_{t})\nonumber\\
&=&-(-W_{t})^\alpha B_{t}.
\end{eqnarray}

Thus, Proposition \ref{th3} holds.
\end{proof}

\par

\bibliographystyle{amsplain}
\end{document}